\documentclass[12pt]{article}

\usepackage{graphicx}
\usepackage{amsmath}
\usepackage{amssymb}
\usepackage{amsthm}

\newcommand{\Tr}{\textrm{Tr}}

\newtheorem{thm}{Theorem}

\begin{document}
\title{Operator extension of strong subadditivity of entropy}
\author{Isaac H. Kim\footnote{Institute of Quantum Information, California Institute of Technology, Pasadena CA 91125, USA}}

\date{\today}

\maketitle
\begin{abstract}
We prove an operator inequality that extends strong subadditivity of entropy: after taking a trace, the operator inequality becomes the strong subadditivity of entropy.
\end{abstract}
\section{Introduction}
There are inequalities satisfied  by Von Neumann entropy, such as nonnegativity and strong subadditivity.\cite{Lieb1973,Nielsen2000} Strong subadditivity plays an important role in quantum information theory: it is essential in virtually all nontrivial coding theorems. Strong subadditivity has also found applications to geometric entropy\cite{Casini2004}, conformal field theory(CFT)\cite{Casini2004a}, and topological entanglement entropy\cite{Grover2011}.

In this paper, we prove a certain operator inequality that extends strong subadditivity. Motivation for studying such inequality comes from recent observation of Li and Haldane(LH).\cite{Li2008} LH conjectured and numerically substantiated that eigenvalues of reduced density matrix contains universal information about the quantum phase. While their result concerns a variational wavefunction for a fractional quantum hall system, author was able to reproduce some of the consequences without resorting to the underlying structure of the wavefunction dictated by the CFT.\cite{Kim2012a}

One important assumption was that certain operator generalization of strong subadditivity holds for general quantum states. Specifically, author conjectured the following inequality.
\begin{equation}
\Tr_{AB}(\rho_{ABC} (\hat{H}_{AB}  + \hat{H}_{BC} - \hat{H}_B - \hat{H}_{ABC})) \geq 0,
\end{equation}
where $\hat{H}_A$ is formally defined as
\begin{equation}
\hat{H}_A = -I_{A^c} \otimes \log (\rho_A).
\end{equation}
Other operators are defined similarly. Whenever logarithm of reduced density matrix appears, tensor product with identity for the rest of the subsystems is implicitly assumed. We prove the conjecture in this paper.

\section{Operator generalization of strong subadditivity}
Since Lieb and Ruskai's seminal result\cite{Lieb1973}, alternative proofs for strong subadditivity have been introduced by several authors.\cite{Petz2003,Nielsen2007,Ruskai2007,Effros2009} In particular, Effros recently presented a proof based on perspective of operator convex function.\cite{Effros2009} Effros basically extended the notion of perspective function from real numbers to operators. Given a function $f$, perspective of $f$ is defined as
\begin{equation}
g(x,t) = f(x/t)t.
\end{equation}

If $f(x)$ is convex, $g(x,t)$ is jointly convex in $x$ and $t$. Main insight of Effros is that similar statement holds for function $f$ that is \emph{operator convex.} To be more precise, he showed the following statement.
\begin{thm}
(Effros 2009) If $f(x)$ is operator convex, and $[L,R]=0$, perspective
\begin{equation}
g(L,R) = f(L/R) R
\end{equation}
is jointly convex in the sense that if $L=cL_1 + (1-c)L_2$ and $R=cR_1 + (1-c)R_2$ with $[L_i,R_i]=0$ $(i=1,2)$, $0\leq c \leq 1$,
\begin{equation}
g(L,R) \leq cg(L_1,R_1) + (1-c) g(L_2, R_2).
\end{equation}
\end{thm}

We choose a matrix algebra $\mathcal{B}(\mathcal{H})$ with an inner product structure of $\langle X,Y\rangle = \Tr(XY^{\dagger})$, where $X,Y$ are $n \times n$ matrices. Following Effros, we choose $L$ and $R$ to be superoperators that multiplies matrix from left or right. For $X \in \mathcal{B}(\mathcal{H})$, $L$ and $R$ are defined as follows.
\begin{align}
L X &= \rho X \nonumber \\
RX &= X \sigma.
\end{align}
$L$ and $R$ commute with each other. One can also show the following relations.
\begin{align}
\log (L) X &= \log (\rho) X  \nonumber \\
\log (R) X &= X\log (\sigma).
\end{align}

Effros' result implies the following statement.
\begin{thm}
\begin{equation}
\Tr_{AB}(\rho_{ABC} (\hat{H}_{AB} + \hat{H}_{BC} - \hat{H}_B - \hat{H}_{ABC})) \geq 0.
\end{equation}
\end{thm}
\begin{proof}
Let $f(x) = x\log x$. Since $f(x)$ is operator convex\cite{Bhatia1997}, $g(L,R)=L \log (L) - L \log (R)$ is jointly convex in $L$ and $R$. Therefore,
\begin{align}
\langle g(L,R)(O), O \rangle = \Tr(\rho \log \rho O O^{\dagger} - \rho O \log \sigma O^{\dagger})\label{eq:quasientropy}
\end{align}
is jointly convex in $L$ and $R$ for all $O \in \mathcal{B}(\mathcal{H})$.\footnote{Joint convexity of Eq.\ref{eq:quasientropy} was originally proved by Petz.\cite{Petz1986}} Choose $\rho= \rho_{ABC}$, $\sigma = \rho_{AB} \otimes \frac{I_C}{d_C}$, $O= I_{AB} \otimes P_C$, where $P_C$ is a projector supported on $C$ and $d_C$ is the dimension of $C$. Note
\begin{equation}
\frac{I_A}{d_A} \otimes \rho_{BC} = \frac{1}{d_A^2} \sum_{i=1}^{d_A^2} U_{A,i} \rho_{ABC} U_{A,i}^{\dagger}
\end{equation}
for a set of unitaries $\{U_{A,i} \}$ that forms an orthogonal basis for $\mathcal{B}(\mathcal{H}_A)$. An example can be found in Reference \cite{Pittenger2000}. Using joint convexity,
\begin{equation}
\Tr(\frac{I_A}{d_A} \otimes \rho_{BC} ( \hat{H}_B- \hat{H}_{BC})P_C) \leq \Tr(\rho_{ABC} (\hat{H}_{AB}-\hat{H}_{ABC})P_C).
\end{equation}
Left hand side of the inequality is equal to $\Tr(\rho_{ABC} (\hat{H}_{B} - \hat{H}_{BC})P_C)$. Since the inequality holds for all $P_C$,
\begin{equation}
\Tr_{AB}(\rho_{ABC} (\hat{H}_{AB} + \hat{H}_{BC} - \hat{H}_B - \hat{H}_{ABC})) \geq 0.
\end{equation}
\end{proof}
One may wish to find a similar inequality when partial trace is restricted to $A$ or $B$. In both cases, the resulting operators are not even Hermitian.

\section{Conclusion}
We have proved an operator extension of strong subadditivity conjectured in Reference \cite{Kim2012a}. While our method was based on the operator convexity of the function $f(x) = x\log x$, Effros' approach can be generalized to other operator convex functions as well. It would be interesting to find applications of these new family of operator inequalities.

{\it Acknowledgements---} This research was supported in part by NSF under Grant No. PHY-0803371, by ARO Grant No. W911NF-
09-1-0442, and DOE Grant No. DE-FG03-92-ER40701. I would like to thank Alexei Kitaev for introducing Effros' work.
\bibliographystyle{unsrt}

\end{document}